\documentclass{article}
\usepackage[utf8]{inputenc}
\usepackage[T1]{fontenc}
\usepackage{amsmath, amssymb, algorithm, algpseudocode}
\usepackage{dsfont}
\usepackage{fancyhdr}
\usepackage[table,xcdraw]{xcolor}
\usepackage{graphicx}
\usepackage{multirow}
\usepackage{endnotes}
\usepackage{float}
\usepackage{vmargin}
\usepackage{subcaption}
\usepackage{import}
\usepackage{stackrel}
\usepackage{comment}
\usepackage{wrapfig}
\usepackage{amsthm}
\usepackage{rotating}
\usepackage[title]{appendix}
\usepackage{listings}
\usepackage{tcolorbox}
\usepackage{soul}
\usepackage{float}
\usepackage{stmaryrd}
\usepackage{tikz}
\usepackage{mleftright}
\mleftright
\usetikzlibrary{matrix}
\usepackage{tikz}
\usetikzlibrary{shapes,arrows,chains, decorations.pathmorphing, decorations.pathreplacing}
\tikzset{quantum/.style={decorate, decoration=snake}}

\newcommand{\qa}{\alpha}
\newcommand{\qb}{\beta}

\newcommand{\qG}{\Gamma}
\newcommand{\qd}{\delta}
\newcommand{\qz}{\zeta}
\newcommand{\ql}{\lambda}
\newcommand{\qe}{\varepsilon}

\newcommand{\qr}{\rho}
\newcommand{\qs}{\sigma}
\newcommand{\qt}{\tau}
\newcommand{\cA}{{\cal A}}
\newcommand{\cB}{{\cal B}}
\newcommand{\cC}{{\cal C}}
\newcommand{\cD}{{\cal D}}
\newcommand{\cH}{{\cal H}}
\newcommand{\cM}{{\cal M}}
\newcommand{\cN}{{\cal N}}
\newcommand{\cK}{{\cal K}}
\newcommand{\cL}{{\cal L}}
\newcommand{\RR}{\mathbb R}
\newcommand{\be}{\begin{equation}}
\newcommand{\ee}{\end{equation}}
\newcommand{\rd}{{\rm d}}
\newcommand{\ket}[1]{\left|#1\right\rangle}		
\newcommand{\bra}[1]{\left\langle#1\right|}

\newcommand{\proj}[1]{\ket{#1}\!\bra{#1}}

\definecolor{darkred}{RGB}{179, 16, 32}

\newcommand{\thegame}{\thesection.\arabic{theorem}}

\newenvironment{game}[1][]{
  \refstepcounter{theorem}
  \begin{tcolorbox}[title=Game \thegame: #1]
}{
  \end{tcolorbox}
}

\newcounter{agame} 
\renewcommand{\theagame}{\Alph{agame}}

\newenvironment{agame}[1][]{
  \refstepcounter{agame}
  \begin{tcolorbox}[title=Game \theagame: #1]
}{
  \end{tcolorbox}
}

\usepackage[colorlinks=true, citecolor=blue, linkcolor=blue, urlcolor=blue]{hyperref}
\usepackage{bbm}
\usepackage{slashed}
\usepackage{enumitem}
\usepackage{authblk}
\usepackage[backend=biber, style=numeric, sorting=none]{biblatex}
\usepackage[normalem]{ulem}
\usetikzlibrary{decorations.pathreplacing,calligraphy}

\usepackage{etoolbox} 

\newenvironment{myfigure}[1][]{%
  \refstepcounter{theorem}
  \begin{figure}[#1]
}{%
  \end{figure}
}


\theoremstyle{plain}
\newtheorem{theorem}{Theorem}[section]

\newtheorem{definition}[theorem]{Definition}

\newtheorem{lemma}[theorem]{Lemma}

\begin{document}
	
\title{Practical Unclonable Encryption with Continuous Variables}
\author{Arpan Akash Ray and Boris \v{S}kori\'{c}}
\affil{Eindhoven University of Technology, The Netherlands }
\date{\today} 
\maketitle

\begin{abstract}
We propose the first continuous-variable (CV) unclonable encryption scheme, 
extending the paradigm of quantum encryption of classical messages (QECM) to CV systems. 
In our construction, a classical message is first encrypted classically and then encoded 
using an error-correcting code. 
Each bit of the codeword is mapped to a CV mode by 
creating a coherent state which is squeezed in the $q$ or $p$ quadrature direction,
with a small displacement that encodes the bit.
The squeezing directions are part of the encryption key.
We prove unclonability in the framework introduced by Broadbent and Lord,
via a reduction of the cloning game to a CV monogamy-of-entanglement game. 

Furthermore, we demonstrate that our scheme can be readily implemented with current technology. By incorporating realistic imperfections such as channel noise and detector inefficiencies, we show that the protocol remains robust under these conditions.

\end{abstract}

\section{Introduction}

The marriage of quantum information theory with cryptography 
has given rise to a wide array of protocols that exploit uniquely quantum phenomena, 
most notably the no-cloning principle \cite{Dieks1982,Park1970,Wootters1982},
to achieve security properties that are unattainable in a classical setting. 
One such concept, {\em unclonable encryption}, harnesses the indivisibility of quantum states to prevent an adversary from copying an encrypted message. 

The term Unclonable Encryption (UE) first appeared in 2003 in a paper by Gottesman \cite{Gottesman2003}.
Alice encrypts a classical message into a quantum state.
A security definition was introduced that essentially states 
``{\it If Bob decides that his decryption is valid, then Eve, given the key, has only negligible information about the plaintext}.''
The security definition was formulated in terms of the trace distance between 
encryptions of different messages, and essentially encompasses tamper detection.
In the same framework, 
Leermakers and \v{S}kori\'{c} devised an UE scheme with key recycling \cite{Leermakers2021}.
Broadbent and Lord \cite{broadbent_et_al2020}
introduced the stronger notions of unclonability and unclonable indistinguishability, in which the splitting of the ciphertext is considered, 
and both receiving partners are malicious.
They constrructed UE in the random oracle model.
Several further UE schemes were 
introduced in \cite{Ananth2021,Leermakers2022,KunduTan2025}, 
and results on the feasibility and limitations of UE were given in \cite{Ananth2022, Majenz2021}.

Until now, UE has been exclusively studied in discrete-variable (DV) quantum systems. 
However,
in the field of Quantum Key Distribution (QKD),
continuous-variable (CV) quantum systems have emerged as an attractive alternative to~DV
\cite{PhysRevA.61.010303, PhysRevA.61.022309, PhysRevA.62.062308,
PhysRevA.63.052311,
PhysRevLett.88.057902, article_Grossmman_03, 10.5555/2011564.2011570, PhysRevLett.93.170504,
garcia2007quantum, leverrier:tel-00451021}.
A key advantage of CV-QKD over its DV equivalent is its practicality, see, e.g., Ref.~\cite{RevModPhys.84.621}. 
In broad terms, CV platforms are simpler to operate and exploit decades of progress in coherent optical-communication technology. 
DV protocols, by contrast, need single-photon detectors and sometimes single-photon sources as well. 
These components remain expensive and technically demanding, especially when photon-number resolution is required. 
CV schemes instead use homodyne or heterodyne detection, which is both cheaper and easier to deploy. 
Furthermore, today’s fibre-optic infrastructure is optimised for low-loss transmission at telecom wavelengths (1310 nm and 1550 nm); 
ideal single-photon emitters at these wavelengths have not yet been realised, and frequency up-conversion introduces additional loss and error.
Beyond QKD, the practical advantages hold more generally for other quantum information processing applications,
and this has fueled substantial interest in translating DV-based cryptographic ideas into the CV domain.

In this paper, we propose the first Unclonable Encryption scheme that works with 
Continuous-Variable states.
We prove that our scheme is \textit{unclonable secure}, in the UE framework of \cite{broadbent_et_al2020}. 
We note that this is strictly weaker than the more general notion of \textit{unclonable indistinguishability}. 
Nevertheless, for the practical objective of a one-time encryption scheme, unclonable security already ensures that no adversary can duplicate or replay ciphertexts without detection, which suffices for our application. 
Furthermore, unclonable indistinguishability allows for the adversary to be entangled with the plaintext, 
which gives too much power to the adversary, since the plaintext is classical. 

It turns out that bringing UE from discrete to continuous variables has a number of nontrivial aspects.
On the construction side, the parameters of the scheme need to be tuned such that both decryptability and unclonability are satisfied.
On the proof-technical side,
we use a series of hybrids in order
to connect the {\em cloning game}, which features in the UE security definition,
to the {\em CV monogamy-of-entanglement game}, for which
an upper bound on the winning probability has recently been proven \cite{Culf2022}.
Furthermore it is necessary to allow imperfect completeness in the CV case,
as there is a non-zero probability of decryption failure caused by 
the probability tails of the CV measurements.

Beyond the immediate theoretical interest, a CV-based UE scheme has potential advantages for future practical quantum networks, especially where CV platforms are more readily integrated with existing optical infrastructure. 

In line with this theme of practicality, we also show that our scheme is correct for experimentally feasible values of squeezing. Finally we analyse our protocol under noisy conditions and show robustness for realistic channel parameters.

The outline of the paper is as follows.
In Section~\ref{sec:prelim} we briefly review CV formalism and  
important definitions and results from the literature.
In Section~\ref{sec:protocol} we present our protocol
and verify that it satisfies the definition of a quantum encryption.
Section~\ref{sec:proof} contains the unclonability proof.
The main result is stated in Theorem~\ref{th:main}.
In section~\ref{sec:prac} we discuss squeezing and robustness under channel/detector imperfections. 
In Section~\ref{sec:discussion} we summarize and discuss future work.

\section{Preliminaries}
\label{sec:prelim}

\subsection{Notation}

We use standard bra-ket notation for quantum states.
Hilbert spaces are written as $\cH$ with a subscript.
E.g.\;we write the Hilbert space of a single CV mode as $\cH_1$.
The notation $\cD(\cH)$ stands for the set of density operators on~$\cH$.
The Hamming weight of a string $s$ is written as $|s|$.

\subsection{Continuous Variables; Gaussian states}
\label{sec:prelimCV}

A `mode' of the electromagnetic vector potential
represents a plane wave solution of the vacuum Maxwell equations
at a certain frequency, wave vector and polarisation.
Associated with each mode there is a creation operator $\hat a^\dagger$
and annihilation operator~$\hat a$.
The linear combinations $\hat q=\frac{\hat a+\hat a^\dagger}{\sqrt2}$
and $\hat p=\frac{\hat a-\hat a^\dagger}{i\sqrt 2}$
are easy-to-observe quantities called {\em quadratures}, and they behave
as the position and momentum operator of a harmonic oscillator.
Measuring a single quadrature is called a \textit{homodyne} measurement.

The {\em Gaussian states} are a special class of CV states;
their Wigner function (quasi density function on the phase space) \cite{RevModPhys.84.621} 
is a Gaussian function of the quadrature variables.
An $N$-mode Gaussian state is fully characterized by a displacement vector $d\in\RR^{2N}$
and $2N\times 2N$ covariance matrix $\qG$.
The corresponding Wigner function is
$\frac1{\pi^N\sqrt{\det \qG}} \exp - (x-d)^T \qG^{-1}(x-d)$,
where $x$ stands for the vector $(q_1,p_1,\ldots,q_N,p_N)^T$.
The class of Gaussian states contains important states like the vacuum, thermal states, coherent states,
squeezed coherent states and EPR states (two-mode squeezed vacuum).
A coherent state that is squeezed in the $q$-direction has covariance matrix
${e^{-\qz}\; 0\;\choose 0 \;\;\;\; e^\qz}$, where $\qz\geq 0$ is the squeezing parameter.
For the $p$-direction it is 
${e^\qz\;\; 0\; \choose 0 \;\; e^{-\qz}}$.
All intermediate directions are possible, but will not be used in this paper.

An EPR state is obtained by mixing a $q$-squeezed vacuum with a $p$-squeezed vacuum using a 50/50 beam splitter.
The resulting two-mode squeezed (TMS) state has zero displacement, and its covariance matrix is
${ \;\mathbb{I}\cosh \qz \;\;\;\; \qs_z \sinh\qz \choose \qs_z \sinh \qz \;\;\; \mathbb{I}\cosh \qz}$, 
where $\mathbb{I}$ is the $2\times2$ identity matrix and $\qs_z={1 \;\;\; 0\choose 0\; -1}$.
In the entanglement based version of our scheme we will make use of a state that resembles
a displaced EPR state, with for instance displacement $d=(\qa,0,\qa,0)^T$.
When the $q$-quadrature is measured on one side, the probability density for
measurement outcome $x_A$ is a normal distribution with mean $\qa$ and variance
$\frac12\cosh\qz$.
The measurement projects the state on the other side to a squeezed coherent state
with displacement $(\qa+[x_A-\qa]\tanh\qz,0)^T$
and covariance matrix ${\frac1{\cosh\qz}\;\;\; 0\;\choose 0 \;\;\;\; \cosh\qz}$.
If the $x_A$ is not known, the displacement $\qa+[x_A-\qa]\tanh\qz$ is a stochastic variable
following a Gaussian distribution with mean $\qa$ and variance $\frac12\cosh \qz \tanh^2 \qz$.

For a comprehensive review of CV quantum information we refer to~\cite{RevModPhys.84.621}.

\subsection{Definitions and useful lemmas}

In Section~\ref{sec:protocol} we will introduce a scheme that 
encrypts a classical message into a quantum ciphertext (cipherstate), using a classical key.
For the formal description of such a scheme we follow the definition given in \cite{broadbent_et_al2020},
with a small modification: we allow for a small probability of failure in the decryption.

\begin{definition}
\label{defQECM}
Let $\lambda\in{\mathbb N}^+$ be a security parameter.
Let ${\cal M}(\ql)$ be the (classical) plaintext space.
Let ${\cal K}(\ql)$ be the (classical) key space.
A {\bf quantum encryption of classical messages (QECM)} scheme is a triplet of efficient quantum circuits 
            $ (\text{Key}_\lambda, \text{Enc}_\lambda, \text{Dec}_\lambda)$ implementing CPTP maps of the form            
        \begin{align}
        \text{Key}_\lambda &: \mathcal{D}(\mathcal{C}) \to \mathcal{D}(\mathcal{H}_{{\cal K}(\lambda)}) \\
        \text{Enc}_\lambda &: \mathcal{D}(\mathcal{H}_{{\cal K}(\lambda)} \otimes \mathcal{H}_{{\cal M}(\lambda)}) \to \mathcal{D}(\mathcal{H}_{T,\lambda}) \\
        \text{Dec}_\lambda &: \mathcal{D}(\mathcal{H}_{{\cal K}(\lambda)} \otimes \mathcal{H}_{T,\lambda}) \to \mathcal{D}(\mathcal{H}_{{\cal M}(\lambda)})
        \end{align}
       where $\cH_{\cK(\ql)}$ is the Hilbert space of the key,
       $\cH_{\cM(\ql)}$ is the Hilbert space of the plaintext,
       and $\cH_{T,\ql}$ is the cipherstate space.
        \noindent
Let ${\rm E}_k$ denote the CPTP map $\rho \mapsto \text{Enc}_\ql(\proj k \otimes \rho)$,
and let ${\rm D}_k$ stand for the CPTP map $\rho \mapsto \text{Dec}_\ql(\proj k \otimes \rho)$.
For all 
$k \in \cK(\ql)$,  $m \in \cM(\ql)$,  it must hold that
\be
    \operatorname{Tr}\big[ \proj k \text{Key}_\ql(1)\big] > 0 
    \quad\implies\quad
    \operatorname{Tr}\Big[\proj m \, {\rm D}_k \circ {\rm E}_k(\proj m)\Big] \geq 1-\qe_{\rm DF}
\ee
where $\qe_{\rm DF}$ is the tolerated probability of decryption failure. 
\end{definition}

Unclonability of a QECM scheme is defined via a {\em cloning game}.
A challenger prepares a cipherstate and gives it to Alice.
Alice splits the cipherstate into two pieces; one piece goes to Bob, one to Charlie.
Then Bob and Charlie receive the decryption key and must {\em both} produce the correct plaintext, without
being allowed to communicate.
A QECM scheme is considered to be secure against cloning if the three players ABC, acting together,
have an exponentially small probability of winning the game.
We follow the definitions of \cite{broadbent_et_al2020} for the cloning game and the security against cloning. 

\begin{definition}[Cloning attack]
\label{def:cloning}
Let $S$ be a QECM scheme, with Hilbert spaces $\cH_{\cK(\ql)}$, $\cH_{\cM(\ql)}$, $\cH_{T,\ql}$ as given
in Def.\,\ref{defQECM}. 
Let $\cH_B$ and $\cH_C$ be Hilbert spaces of arbitrary dimension.
A cloning attack against $S$ is a triplet of efficient quantum circuits $(\mathcal{A}, \mathcal{B}, \mathcal{C})$ implementing CPTP maps of the form
        \begin{align}
           \mathcal{A} &: \mathcal{D}(\mathcal{H}_{T,\lambda}) \to \mathcal{D}(\mathcal{H}_{B} \otimes \mathcal{H}_{C}),\\
           \mathcal{B} &: \mathcal{D}(\mathcal{H}_{{\cal K}(\lambda)} \otimes \mathcal{H}_{B}) \to \mathcal{D}(\mathcal{H}_{\cM(\lambda)}), \, \text{and}\\
           \mathcal{C} &: \mathcal{D}(\mathcal{H}_{\cK(\lambda)} \otimes \mathcal{H}_{C}) \to \mathcal{D}(\mathcal{H}_{\cM(\lambda)}).
        \end{align}
\end{definition}

\begin{definition}
\label{def:unclonable}
Consider the cloning attack according to Def.\,\ref{def:cloning}.
Let $B_k$ stand for the CPTP map $\qr\mapsto \cB(\proj k\otimes \qr)$, and analogously $C_k$.
Let $\cM(\ql)=\{0,1\}^n$.
A QECM scheme $S$ is $\qt(\lambda)$-{\bf uncloneable secure} if for all cloning attacks $(\cA,\cB,\cC)$ against $S$, there exists a negligible function $\eta$ such that
        \begin{equation}\label{security conditon}
        \mathbb{E}_{m, k} \operatorname{Tr}\Big(\big(\proj m \otimes \proj m\big) \big(B_k \otimes C_k\big) \circ \cA \circ 
        \text{E}_k\big(\proj m\big)\Big) 
        \leq 2^{-n + \qt(\lambda)} + \eta(\lambda).
        \end{equation}
Here the expectation ${\mathbb E}$ is over uniform $m$, 
and $k$ distributed according to $\Pr[K = k] = \operatorname{Tr}\big[\proj k Key_\ql(1)\big]$.
If $S$ is 0-uncloneable secure, we simply say that it is uncloneable secure.
\end{definition}

Our security proof will make use of recent results 
on {\em entanglement monogamy games} for CV systems \cite{Culf2022}.
Such a game is played between Alice on one side and Bob and Charlie on the other.
Bob and Charlie prepare a tripartite state $\rho_{ABC}$ of their choice.
Then Alice does an unpredictable measurement, and Bob and Charlie have to show that they are both 
sufficiently entangled with Alice to guess Alice's outcome to some degree of accuracy.
Below we present the
\textit{coset monogamy game on $\mathbb{R}^N$}
as introduced in Section 4 of \cite{Culf2022}. 
 In the original version, the measurement bases are continuous variable coset states of \(\mathbb{R}^N\), given by
$
   |P_{q,p}\rangle = \bigotimes_{i=1}^N \begin{cases} |q = q_i\rangle, & i \notin I,\\ |p = p_i\rangle, & i \in I. \end{cases}
$
Here, $I$ is a subset interval of $N$. 
We note that this is equivalent to performing a homodyne measurement in the $q$ direction if $i\in I$,  otherwise in the $p$ direction.
We refer to this game as the `partial' CV monogamy of entanglement game,
since it requires Bob and Charlie to produce guesses for only one quadrature direction 
instead of two.
In Section~\ref{sec:fullmonogamy} we introduce a `full' version of the game.

\begin{game}[CV Partial Monogamy of Entanglement Game]
\label{game:PartialMonogamy}

\begin{enumerate}
    \item \textbf{Initial state preparation.} \\
    Bob and Charlie prepare an $M$-mode state 
    \(\rho_{ABC}\) across three registers: one for Alice, one for Bob, and one for Charlie.
    After preparation, they are no longer allowed to communicate.
    
\item \textbf{Alice’s measurement choice and outcomes.} \\
Alice chooses quadrature directions $(\theta_i)_{i=1}^M$,
\(\theta_i \in \{0, \tfrac{\pi}{2}\}\), such that 
$|i: \; \theta_i=0|=M/2$, i.e.\;each direction is chosen exactly $M/2$ times.
She then does homodyne detection of mode $i$ in direction $\theta_i$,
getting outcomes that we denote as
    \[
       q_i \quad \text{if } \theta_i = 0, \quad\quad p_i \quad \text{if } \theta_i = \tfrac{\pi}{2}.
    \]
    
\item \textbf{Announcement and responses.} \\
Alice announces $(\theta_i)_{i=1}^M$.
Bob outputs $q_B\in\RR^{M/2}$ containing his guesses for Alice's $q$-values.
Charlie outputs $p_C\in\RR^{M/2}$ containing his guesses for Alice's $p$-values.
    
\item \textbf{Winning condition.} \\
Bob and Charlie win the game if 
    \begin{equation}
          \|\,q - q_B \|_\infty < \delta 
       \quad \text{and} \quad 
       \|\,p - p_C \|_\infty < \varepsilon.
    \end{equation}
\end{enumerate}
\end{game}
The $\infty$-norm in the winning condition means that all guesses $(q_B)_i$ and $(p_C)_i$
must be close to Alice's values.

\noindent
As shown in \cite{Culf2022}, this game can viewed as an \emph{abelian coset measure} monogamy game, where Bob and Charlie attempt to guess the measurement outcomes in each quadrature within error neighborhoods \( (-\delta,\delta)^N\) and \( (-\varepsilon,\varepsilon)^N\). 
The following upper bound was obtained on the winning probability.

\begin{lemma}
\label{lemma:monogamy-bound} 
(Theorem 4.1 in~\cite{Culf2022}).
In the CV Partial Monogamy of Entanglement Game (\ref{game:PartialMonogamy}), the winning probability $w$ is upper bounded as
\begin{equation}
\label{eq:monogamy-bound}
  w
  \;\le\; 
  \frac{1}{\binom{N}{N/2}}
  \sum_{k=0}^{N/2} 
     \binom{N/2}{k}^2 
     \bigl(2\sqrt{\delta\,\varepsilon}\bigr)^k
  \;\le\;
  \sqrt{e}\bigl(\tfrac12 + \sqrt{\delta\,\varepsilon}\bigr)^{\tfrac{N}{2}}\,.
\end{equation}
\end{lemma}

\section{Protocol Description}
\label{sec:protocol}

We propose a QECM scheme that makes use of squeezed coherent states.
The scheme encrypts a message $m\in\{0,1\}^n$ to an $N$-mode cipherstate.
First we apply a symmetric classical encryption scheme;
this step ensures message confidentiality independent of the unclonability.
We do not specify which symmetric cipher is used.
We denote the encryption and decryption operations as
${\tt Enc}_{\rm base}(\cdot,\cdot)$ and
${\tt Dec}_{\rm base}(\cdot,\cdot)$.

Then the classical ciphertext gets encoded into an $N$-bit codeword~$c$.
The error correcting code is chosen such that it that can correct $t$ bit errors.
The encoding and decoding algorithms are denoted as {\tt Encode}, {\tt Decode}.

We encode a bit value $c_i$ into a CV mode by displacing a squeezed vacuum state  
over a distance $(-1)^{c_i} \qa$, where $\qa>0$ is a constant.
The displacement is in the `narrow' direction of the state. 
The underlying idea is that the squeezing direction is part of the encryption key;
the attacker $\cA$, who does not know this direction, has trouble determining $c$ and hence
cannot create good clones. 
This encoding is similar to conjugate coding for qubits \cite{Wiesner1983}.

For proof-technical reasons, the squeezing directions are not chosen uniformly at random.
We impose the constraint that exactly half the modes are squeezed in the $q$-direction,
and one half in the $p$-direction.
Hence the corresponding key space is not $\{0,1\}^N$ but rather the set
$\cL=\{1,\ldots,\log_2{N\choose N/2}\}$ of labels which uniquely enumerate
the strings with Hamming weight~$N/2$.

The QECM algorithms (Key, Enc, Dec) are given below.

\begin{algorithm}[H]
\caption{Key Generation (Key)}
\label{alg:Key}

\textbf{Input:} $z(\lambda)$, $N(\lambda)$, $r(\lambda)$, $\qa(\lambda)$.\\
 \textbf{Output:} a symmetric key $s\in\{0,1\}^z$, a binary string $\phi\in\{0,1\}^N$ 
        and a real vector $k\in\mathbb{R}^N$.
 \begin{algorithmic}[1]
\State Sample $s$ uniformly from $\{0,1\}^z$.
\State Sample label $\ell$ uniformly from $\{1,\ldots,\log_2{N\choose N/2}\}$.
Convert $\ell$ to a string $\phi\in\{0,1\}^N$ with Hamming weight~$N/2$.
\For{\( i = 1 \) to \( N \)}
\State Sample $k_i$ from the normal distribution $\cN(0,\frac12 \cosh r\tanh^2 r)$ truncated
to the interval 
$(-\qa\tanh r ,\qa\tanh r )$.

\EndFor
\State $k=(k_1,\ldots,k_N)$.
\State Output \( (s, \phi, k) \).
\end{algorithmic}
\end{algorithm}

Note that the pdf of $k_i$ is proportional to $\exp(-\frac{k_i^2}{\cosh r\tanh^2 r})$, i.e.\;Gaussian form,
but the support $k_i\in(-\qa\tanh r, \qa\tanh r)$ is quite narrow compared to the width of the Gaussian.

\begin{algorithm}[H]
\caption{Encryption (Enc)}\label{alg:Enc}

\textbf{Input:} 
Key $(s, \phi, k) \in \{0,1\}^z\times\{0,1\}^N\times\mathbb{R^N}$;
message $m\in\{0,1\}^n$;  
parameters $\qa,r \in \mathbb{R}^+$.\\
\textbf{Output:} Cipherstate  $\qr\in\cD(\cH_1^{\otimes N})$.
\begin{algorithmic}[1]
\State Compute ciphertext $m'={\tt Enc}_{\rm base}(s,m)$.
\State Compute codeword $c={\tt Encode}(m')$.
\For{\( i = 1 \) to \( N \)}
\State Prepare single mode squeezed state $\qr_i$ with displacement $d_i$ and covariance matrix $\qG_i$,
where $d_i=[\qa(-1)^{c_i}+k_i] {\overline{\phi_i}\choose \phi_i}$, and 
$
\qG_i=
\begin{pmatrix}
(\cosh r)^{2\phi_i-1} & 0 \\
0 & (\cosh r)^{1-2\phi_i}
\end{pmatrix}.
$

\EndFor
\State $\qr=\qr_1\otimes\cdots\otimes\qr_N$.
\State Output $\qr$.

\end{algorithmic}
\end{algorithm}

\begin{algorithm}[H]
\caption{Decryption (Dec)}
\label{alg:Dec}

 \textbf{Input:} 
 Key $(s, \phi, k) \in \{0,1\}^z\times\{0,1\}^N\times\mathbb{R^N}$;
 cipherstate  $\qr\in\cD(\cH_1^{\otimes N})$. \\
\textbf{Output:} Message $\hat m\in\{0,1\}^n$.
\begin{algorithmic}[1]
\For{\( i = 1 \) to \( N \)}
\State Perform homodyne measurement on $\qr_i$ in the $\phi_i \cdot \frac\pi2$ direction, 
resulting in outcome $y_i\in\RR$.
\State
$\hat c_i=\frac12-\frac12\,{\rm sign}(y_i-k_i)$.
\EndFor
\State $\hat c=(\hat c_1,\ldots,\hat c_N)$.
\State $\mu={\tt Decode}(\hat c)$.
\State $\hat m={\tt Dec}_{\rm base}(s,\mu)$.
\State Output $\hat m$.

\end{algorithmic}
\end{algorithm}

\begin{theorem}
\label{th:epsDF}
Our construction is a QECM scheme with decryption failure parameter
\begin{equation}
    \qe_{\rm DF}=\exp\Big[-N D_{\rm KL}(\frac{t+1}N|| \frac12{\rm Erfc}(\qa \sqrt{\cosh r}) ) \Big],
\label{epsDF}
\end{equation}
where $D_{\rm KL}$ stands for the binary Kullback-Leibler divergence,
$D_{\rm KL}(a||b)=a\ln\frac ab+(1-a)\ln\frac{1-a}{1-b}$,
and Erfc is the complementary error function.
\end{theorem}

\begin{proof}
It is trivial to see that the triplet (Key, Enc, Dec) fits the format of the CPTP maps in Def.\,\ref{defQECM}. 
All that is left to show is that the protocol satisfies the correctness condition.
A bit error in the codeword bit $c_i$ occurs when $y_i-k_i$ has the wrong sign.
Consider, without loss of generality, $c_i=0$.
Then the bit error probability is $\qb:={\rm Pr}[Y_i-k_i<0]$.
When $k_i$ is known,
the random variable $Y_i-k_i$ is gaussian-distributed with mean $\qa$ and variance
$1/(2\cosh r)$.
Hence 
\begin{equation}
    \qb = \frac12-\frac12{\rm Erf}(\qa \sqrt{\cosh r}).
\label{ber}
\end{equation}
A decoding error can occur only when there are more than $t$ bitflips; 
this occurs with probability 
$P=\sum_{j=t+1}^N {N\choose j}\qb^j (1-\qb)^{N-j}$
$=\sum_{\ell=0}^{N-t-1}{N\choose \ell} \qb^{N-\ell}(1-\qb)^\ell$.
Using a Chernoff bound on the binomial tail, we obtain
$P \leq \exp\Big[-N D_{\rm KL}(\frac{t+1}N||\qb) \Big]$.
\end{proof}
Note that asymptotically $N$ gets close to $\frac n{1-h(\qb)}$, where $h$ is the binary entropy function. 

\section{Proving Unclonability}
\label{sec:proof}

\subsection{Game hopping}

We prove that our scheme satisfies the security definition of Def.\,\ref{def:unclonable}.
We do this as follows.
(i)~We rewrite the state preparation from the original prepare-and-send form to Entanglement Based (EB) form.
(ii)
We show, in a series of `hops' (hybrids), that the cloning game that the security definition is based on
is equivalent to a CV entanglement monogamy game. 
(iii)
Finally we use Lemma~\ref{lemma:monogamy-bound} which upper bounds the winning probability in the CV entanglement monogamy game.

The sequence of games is as follows
\begin{itemize}
\item
The cloning game for the actual QECM scheme.
\item
The cloning game for the Entanglement Based form of the QECM scheme.
\item
A variant of the above game, where now the keys $s$ and $k$ are {\em not} revealed to Bob and Charlie. 
Only the squeezing directions $\phi$ are revealed.
\item
The Full CV entanglement monogamy game.
\item
The Partial CV entanglement monogamy game (Game~\ref{game:PartialMonogamy}).
\end{itemize}

\subsection{Entanglement based version}
\label{sec:EB}

The first hop is to replace (in the encryption algorithm) the 
drawing of $k_i$ and the
preparation of the single-mode state $\qr_i$ 
by the following procedure.
(1)
Prepare a two-mode entangled state $\qr_{Ch A}^{(i)}$.
(2)
Do a homodyne measurement on the `Ch' subsystem in the $\phi_i$ direction, yielding outcome~$u_i$.
Compute $k_i=[u_i-\qa(-1)^{c_i}]\tanh r$.

The state $\qr_{Ch A}^{(i)}$ must be such that it yields the correct distribution for $k_i$ 
(truncated Gaussian)
when the Ch system is measured in the $\phi_i$ direction,
{\em and}
the state of the `A' subsystem gets projected to the correct squeezed coherent state.
Without the truncation, $\qr_{Ch A}^{(i)}$ would simply be
given by
the displaced EPR state mentioned in Section~\ref{sec:prelimCV}.
In order to reproduce the truncated $k_i$-interval, however, 
the EPR state needs to be modified.
The details are presented in Appendix~\ref{app:state}.

\subsection{Unclonable Encryption game}
The cloning attack (Def.\,\ref{def:cloning}) can be represented as the game below,
between the Challenger and the three players $\cA\cB\cC$.

\begin{agame}[Unclonable Encryption]
\label{game:cloning game}
    \begin{enumerate}
    \item \textbf{Initial state preparation.} \\
The challenger Ch picks a random message~$m$.\\
    \uline{Prepare-and-send}:  
$Ch$ runs Key (algorithm~\ref{alg:Key}) and obtains $(s,\phi,k)$.
He runs Enc (algorithm~\ref{alg:Enc}) on~$m$.
He sends the resulting cipherstate to Alice.\\
    \uline{Entanglement-based}:  
$Ch$ runs part of algorithm~\ref{alg:Key} and obtains $(s,\phi)$.
    He prepares an entangled state $\rho_{Ch A}=\bigotimes_{i=1}^N\rho_{Ch A}^{(i)}$, 
    as explained in Section~\ref{sec:EB}.
    He sends the `A' part to Alice.
    
    \item \textbf{Distributing quantum information to co-players.} \\
     Alice ($\mathcal{A}$), "splits" her quantum state and sends the two parts to her co-players, Bob ($\mathcal{B}$) and Charlie ($\mathcal{C}$). 
     After that $\mathcal{A}, \mathcal{B}$ and $\mathcal{C}$ are no longer allowed to communicate. 
    
    \item \textbf{Key opening and response.} \\
    \uline{Prepare-and-send}: 
    $Ch$ announces the key $(s, \phi, k)$. 
    $\mathcal{B}$ and $\mathcal{C}$ respond with $m_B$ and $m_C$ respectively.\\
    \uline{Entanglement-based}:  $Ch$ does homodyne measurements on his modes and computes $k$ from the outcomes. 
    $Ch$ announces the key $(s, \phi, k)$. 
    $\mathcal{B}$ and $\mathcal{C}$ respond with $m_B$ and $m_C$ respectively.
    \item \textbf{Winning condition.} \\
    The triplet of players $\cA\cB\cC$ win the game if $m_B=m_C=m$.
    \end{enumerate}
    \end{agame}

\begin{myfigure}[ht]
\centering
\begin{tikzpicture}[scale=0.75, every node/.style={transform shape, font=\large},
  box/.style={rectangle, draw=black, rounded corners, minimum width=2.5cm, minimum height=1.2cm, text centered},
  arrow/.style={->, >=latex, thick},
  node distance=1.5cm
]

  \node[box] (Ch) at (0,0) {${Ch}$};
  \node[box] (A) [right=4cm of Ch] {$\mathcal{A}$};

  \node[box] (B) [above right=0.5cm and 2cm of A] {$\mathcal{B}$};
  \node[box] (C) [below right=0.5cm and 2cm of A] {$\mathcal{C}$};

  \node (keyB) [right=1.5cm of B] {key};
  \node (msgB) [below=0.4cm of B] {$m_B$};
  \node (keyC) [right=1.5cm of C] {key};
  \node (msgC) [below=0.4cm of C] {$m_C$};

  \draw[arrow] (Ch) -- node[above, midway, font=\large] {Quantum ciphertext} (A);
  \draw[arrow] (A) -- node[above, font=\large] {} (B);
  \draw[arrow] (A) -- node[above, font=\large] {} (C);
  \draw[arrow] (keyB) -- node[above, font=\large] {} (B);
  \draw[arrow] (B) -- node[above, font=\large] {} (msgB);
  \draw[arrow] (keyC) -- node[above, font=\large] {} (C);
  \draw[arrow] (C) -- node[above, font=\large] {} (msgC);
\end{tikzpicture}
\caption{\it Schematic representation of the unclonable encryption game.}
\label{fig:protocol}
\end{myfigure}

Note that in this game the measurement by Ch has been postponed to step~3,
as opposed to step~1 in the direct EB description of the preparation of A's state.

    \begin{agame}[EB Intermediate Unclonable Encryption]
    \label{game:int cloning game}
    \begin{enumerate}
    \item \textbf{Initial state preparation.} \\
    $Ch$ prepares the entangled state $\rho_{Ch A}$. He sends one half of the state to Alice.
    
    \item \textbf{Distributing quantum information to co-players.} 

\item\textbf{Key opening (only $\phi$) and response}.\\
$Ch$ measures $u$ and announces $\phi$.  
$\cB$ and $\cC$ respond with $u_B$ and $u_C$ respectively. (Their guess for $u$.)

\item \textbf{Winning condition.}\\
For $i\in\{1,\ldots,N\}$ $Ch$ calculates 
$c_{Bi}=\frac12-\frac12{\rm sign}(u_{Bi}-\frac{k_i}{\tanh r})$
and
$c_{Ci}=\frac12-\frac12{\rm sign}(u_{Ci}-\frac{k_i}{\tanh r})$.
Then $m_B={\tt Dec}_{\rm base}(s, {\tt Decode}(c_B))$
and $m_C={\tt Dec}_{\rm base}(s, {\tt Decode}(c_C))$.
The players $\cA\cB\cC$ win if $m_B=m_C=m$.

    \end{enumerate}
\end{agame}

\begin{lemma}
\label{lemma:UE games eq}
Winning 
the Entanglement Based Unclonable Encryption game (Game \ref{game:cloning game})
is equivalent to winning 
the EB Intermediate Unclonable Encryption game (Game \ref{game:int cloning game}) .
\end{lemma}

\begin{proof}
Consider the EB version of Game~\ref{game:cloning game}.
Given $k_i$, Bob and Charlie know that there are only two possible values for $u_i$, namely $u_i=\frac{k_i}{\tanh r}\pm\qa$.
Hence, getting a bit $c_i$ correct is equivalent to making the correct binary choice for $u_i$.
From the fact that, by construction, $u_i\in(-2\qa,2\qa)$, it follows that this in turn is equivalent
to determining $u_i$  with a resolution of $\qa$ or better.

\noindent
In Game~\ref{game:int cloning game} it is precisely such resolution on the $u$-axis that is required to guess~$c_i$. 
\end{proof}

\subsection{Monogamy of entanglement game}\label{sec:fullmonogamy}

We introduce a version of the entanglement monogamy game that is 
closer to the encryption scheme than Game~\ref{game:PartialMonogamy}.

\begin{agame}[CV Full Monogamy of Entanglement game]
\label{game:FullMonogamy}
\begin{enumerate}
    \item \textbf{Initial state preparation.} \\
    Bob and Charlie prepare an $M$-mode state 
    \(\rho_{ABC}\) across three registers: one for Alice, one for Bob, and one for Charlie.
    After preparation, they are no longer allowed to communicate.
    
\item \textbf{Alice’s measurement choice and outcomes.} \\
Alice chooses quadrature directions $(\theta_i)_{i=1}^M$,
\(\theta_i \in \{0, \tfrac{\pi}{2}\}\), such that 
$|i: \; \theta_i=0|=M/2$, i.e.\;each direction is chosen exactly $M/2$ times.
She then does homodyne detection of mode $i$ in direction $\theta_i$,
getting outcomes that we denote as
    \[
       q_i \quad \text{if } \theta_i = 0, \quad\quad p_i \quad \text{if } \theta_i = \tfrac{\pi}{2}.
    \]
    
\item \textbf{Announcement and responses.} \\
Alice announces $(\theta_i)_{i=1}^M$.
Bob outputs a list $q_B\in\RR^{M/2}$ containing his guesses for Alice's $q$-values
and a list $p_B\in\RR^{M/2}$ containing his guesses for Alice's $p$-values.
Similarly, Charlie outputs $q_C$, $p_C$.
    
    \item \textbf{Winning condition.} \\
    Bob and Charlie win the game if 
    \begin{equation}\label{strgame winning con}
          \|\,q - q_B \|_\infty < \delta,
          \quad 
          \|\,p - p_B \|_\infty < \delta,
          \quad 
          \|\,q - q_C \|_\infty < \varepsilon,
          \quad \text{and} \quad 
          \|\,p - p_C \|_\infty < \varepsilon.
    \end{equation}
\end{enumerate}
\end{agame}

\begin{lemma}
\label{lemma:reduction}
Let $w^{\rm UE}(N,t)$ denote the winning probability for Game~\ref{game:int cloning game}.
Let $w^{\rm full}(N-2t)$ be the winning probability for Game~\ref{game:FullMonogamy}
with  $\qd=\qa$, $\qe=\qa$, $M=N-2t$.
It holds that
\be
    w^{\rm UE}(N,t) \leq 2^{N-n}\; w^{\rm full}(N-2t).
\label{reduction}
\ee
\end{lemma}

\begin{proof}
In Game~\ref{game:int cloning game}, getting a bit $c_i$ correct is equivalent to getting $u_i$ correct within a distance~$\qa$. 
This corresponds to setting $\qd=\qa$ and $\qe=\qa$ in Game~\ref{game:FullMonogamy}. 

Next, in Game~\ref{game:int cloning game} it is required to get at least $N-t$ bits of $c$ correct.
Let $e\in\{0,1\}^N$ denote an error pattern.
Let $w^{\rm UE}_e$ be the probability of winning with precisely error pattern~$e$.
The bit errors can be arbitrarily distributed over the $q$ part and the $p$ part, which does not nicely fit the
symmetric structure of the monogamy game.
In order to obtain the symmetric structure we loosen our requirements a little, and allow
a surplus of bit errors in one block (if any) to be balanced by additionally allowed bit errors in the other block.
In the worst case, all bit errors are located in one block;
this leads to an allowed $2t$ errors. 
We write
\be
    w^{\rm UE}(N,t) = \!\!\!\!
    \sum_{e\in\{0,1\}^N: |e|\leq t} \!\!\!\!\!\!\!\!  w^{\rm UE}_e
    \leq \sum_{e\in\{0,1\}^N: |e|\leq t} \!\!\!\!\!\!\!\! w^{\rm full}(N-2t)
    = w^{\rm full}(N-2t) \; \Big|e\in\{0,1\}^N: |e|\leq t\Big|.
    \quad
\ee
Finally we apply the Hamming bound for binary codes.
\end{proof}
\subsection{Main theorem}
\label{sec:maintheorem}

\begin{theorem}
\label{th:main}
Our QECM scheme is unclonable secure according to Def.\,\ref{def:unclonable}, with
\be
    \qt(\ql)=  \frac N2 +\left(\frac N2-t\right)\log(1+2\qa) +t+\frac12\log e.
\label{result}
\ee
\end{theorem}

\begin{proof}
The cloning attack in Def.\,\ref{def:unclonable} has a success probability
equal to the winning probability in Game~\ref{game:cloning game}.
By Lemma~\ref{lemma:UE games eq}, this equals the winning probability of Game~\ref{game:int cloning game}.
Next, by Lemma~\ref{lemma:reduction} this is upper bounded by
$2^{N-n}\; w^{\rm full}(N-2t)$.
We use the fact that the full monogamy game \ref{game:FullMonogamy} is harder (or equally hard) to win than the partial monogamy game \ref{game:PartialMonogamy}.
Thus we have
$
    w^{\rm full}(N-2t) \leq \sqrt e (\frac12+\qa)^{\frac N2-t}
$.
\end{proof}

{\em Remark}: For $n\gg 1, \qa\ll 1$ it holds that $\qt(\ql)\approx \frac N2$.
This is a large number; however, it still results in a winning probability that is 
exponentially small in~$n$, which is the main objective for a QECM scheme.

We see from the result (\ref{result}) that $\qa$ needs to be smaller than approximately~$\frac12$,
otherwise the attackers' winning probability is not exponentially small in~$n$.
This is unsurprising, for the following reason.
One obvious cloning attack on our scheme is to perform a \textit{heterodyne} measurement on each cipherstate mode individually.
This corresponds to mixing the squeezed state with the vacuum state using a 50/50 beam splitter.
The effect is a deterioration of the signal-to-noise ratio:
On the one hand, the signal power $\qa^2$ is reduced to $\qa^2/2$.
On the other hand, the noise power goes 
from $\frac1{2\cosh r}$ to $\frac12(\frac12+\frac1{2\cosh r})$
due to the averaging with the vacuum's shot noise, which has power $\frac12$.
The resulting signal-to-noise ratio in the heterodyne measurement is 
$2\qa^2/(1+\frac1{\cosh r})$
$\approx 2\qa^2$. 
Hence at $\qa={\cal O}(1)$ there is significant leakage about the codeword bits~$c_i$,
putting the unclonability at risk.

We consider the asymptotic case $n\to\infty$.
In this limit the code rate approaches capacity and we write $n\approx N[1-h(\qb)]$.
From the security definition we have the requirement that $2^{-n+\qt}$ should be asymptotically small.
Substituting $n\to N[1-h(\qb)]$ in (\ref{result}), the requirement $-n+\qt < 0$ can be rewritten as
\be
    h(\qb) - (\frac12-\qb)[1-\log(1+2\qa)] < 0.
\label{allowedasymp}
\ee
As $\qb$ is a function of $\qa$ and $r$ (\ref{ber}), eq.(\ref{allowedasymp}) specifies a region in the
$(\qa,r)$-plane where it is possible to achieve security asymptotically.

\begin{figure}[h]
    \centering
    \includegraphics[width=0.6\linewidth]{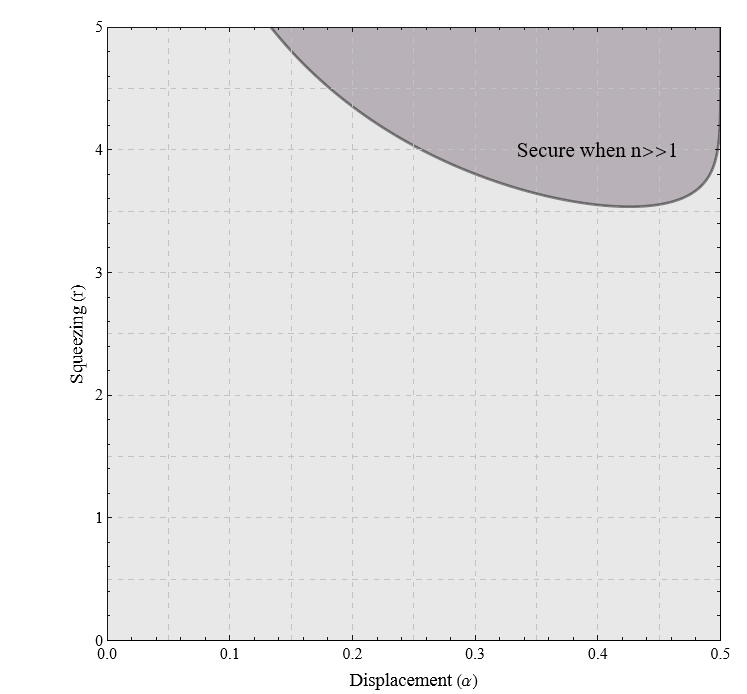}
    \caption{ \textit{We plot (\ref{allowedasymp}) as a function of the squeezing $r$ and displacement $\alpha$. As illustrated, asymptotically the scheme is secure only for $r>3.5$ for some values of $\alpha$. }}
    \label{fig:plot1}
\end{figure}

\section{Practical considerations}
\label{sec:prac}
Since our main motivation for constructing CV-UE is to have a practically implementable scheme,
we examine how the protocol behaves under realistic parameter choices and noise.

Generating highly squeezed light is experimentally challenging. We therefore restrict $r$ to values that are reasonably close to the current state of the art. 
A $15$ dB squeezed vacuum has been demonstrated \cite{Vahlbruch2016}, and $10$ dB is common.\footnote{
10 dB corresponds to $e^r=10$, $r\approx 2.30$;
15 dB corresponds to $e^r=10^{1.5}\approx 31.6$, $r\approx 3.45$.
}

The bit error rate (BER) $\qb$ is given by (\ref{ber}).
We aim for a small BER, which in turn yields a small decryption–failure probability $\qe_{\rm DF}$ (\ref{epsDF}). 
Fig.~\ref{fig:betaplot} displays $\beta$ as a function of $r$, at $\qa=0.4$. 

As an example, consider $\qa=0.40$, $r=3.4$, yielding $\qb= 0.014$. 
We take a ciphertext of size $N=1000$, and we set $t=35$ so that $t>N\qb$.
Then (\ref{epsDF}) gives $\qe_{\rm DF}=6.9\cdot 10^{-6}$. Note that the actual message size in this case will be close to $N(1-h(\beta))\approx892$. We will continue to use this relationship between $N$ and $n$ in the remainder of this section.

\begin{figure}[h]
\centering
\begin{subfigure}{0.48\linewidth}
\hspace*{-0.8cm}
  \centering
  \includegraphics[width=1.15\linewidth]{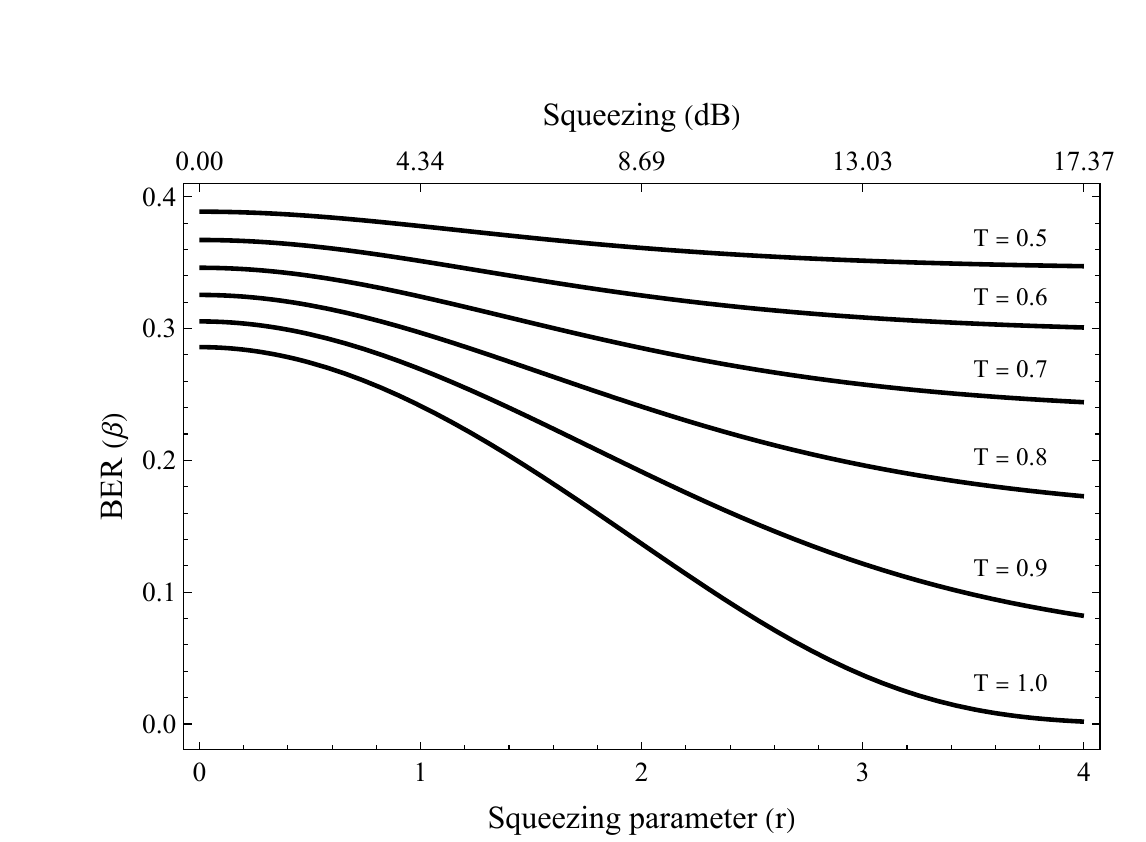}
  \caption{\textit{$\beta_{\mathrm{noisy}}$  as a function of squeezing parameter $r$, at displacement $\alpha=0.4$ plotted for various transmittance $T$. Excess noise $\xi=0.001$ for all curves.}}
  \label{fig:betaplot}
\end{subfigure}
\hfill
\begin{subfigure}{0.5\linewidth}
  \centering
  \includegraphics[width=\linewidth]{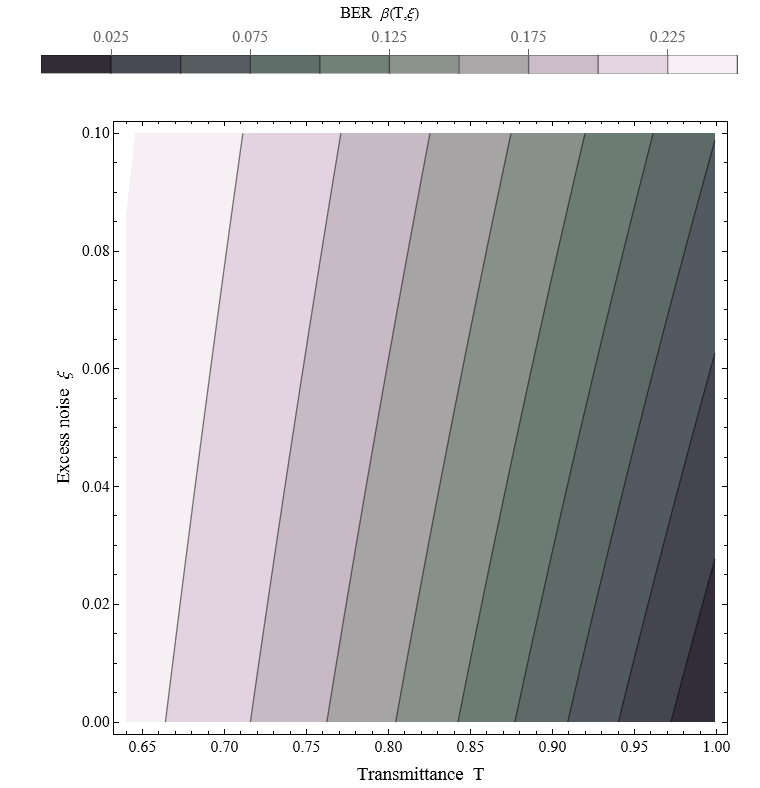}
  \caption{\textit{$\beta_{\mathrm{noisy}}$ as a function of channel transmittance $T$ and excess noise $\xi$. Parameters: $r=3.6$, $\alpha=0.4$.}}
  \label{fig:noiseplot}
\end{subfigure}

\caption{\textit{Bit error rate (\ref{eq:ber_noisy}) for different parameters (channel parameters in SNU).}}
\label{fig:noiseplot_betaplot}
\end{figure}

\subsection{Robustness under noise}
\label{subsec:noise}
We model the optical channel by its transmittance $T\in[0,1]$ and (input-referred) excess-noise power $\xi$ in SNU \cite{leverrier:tel-00451021}. In this model, the receiver’s effective quadrature variance and displacement transform as
\begin{equation}
    \frac{1}{\cosh r}\longrightarrow \frac{T}{\cosh r} + (1-T) + T\xi,
    \qquad
    \alpha\longrightarrow T\alpha,
\label{eq:channel_update}
\end{equation}
which follows from the standard covariance–matrix transformation of a thermal–loss (attenuator) Gaussian channel and linear detection \cite{garcia-patron-sanchez2007thesis}. Detector inefficiency and electronics noise can be absorbed into the overall transmittance/noise budget in this framework \cite{YMKSK2023}.
Under the map \eqref{eq:channel_update}, the BER becomes
\begin{equation}
    \beta_{\mathrm{noisy}}
    = \tfrac12\mathrm{Erfc}\left(
    \frac{T\alpha}{\sqrt{\dfrac{T}{\cosh r} + (1-T) + T\xi}}
    \right).
\label{eq:ber_noisy}
\end{equation}
This expression is plotted in Fig.~\ref{fig:noiseplot} as a function of $(T,\xi)$ for fixed $(r,\alpha)$.

As an example, consider 5 km of standard single-mode optical telecom
fibre at $1550$ nm wavelength. The loss rate is $0.22$ dB/km;
at $5$ km this corresponds to $T\approx0.8$.
At $15$ dB squeezing ($r\approx 3.5$), Fig~\ref{fig:betaplot} then gives a BER around $20$\%, which can be handled with a low-rate code, e.g. an
LDPC code, or a higher-rate code concatenated with a repetition code.

Note that we assume that the adversaries are not using the noisy channel; they have $T=1$ and $\xi=0$.
The channel properties hence need not enter the security analysis.

\section{Discussion}
\label{sec:discussion}

We have introduced the first Unclonable Encryption scheme with Continuous-Variable states,
and have given a proof of unclonability in the game-based framework of~\cite{broadbent_et_al2020}. 
To the best of our knowledge, this is also the first analysis of an unclonable encryption protocol in the presence of realistic channel noise, marking a significant step toward practical quantum cryptography.

Making a meaningful comparison with prior work is difficult, since DV–CV performance cannot be fairly assessed by directly comparing the winning probabilities in the cloning game. The most informative benchmark is probably the DV analogue of our scheme, which can be realized via Wiesner’s conjugate coding \cite{Wiesner1983}.
Fig.\;\ref{fig:plot} plots the winning probability of our UE game as a function of message size, including also
the UE game for conjugate coding into qubit states \cite{broadbent_et_al2020}. 
The conjugate coding scheme first one-time pads a message $m$ 
and then bitwise encodes the ciphertext into the standard basis or the Hadamard basis.
For experimentally feasible squeezing of $r=3.6 (\approx 15.6 \;\mathrm{dB})$, we observe exponentially decreasing probability of the adversaries winning only for message longer message sizes in the order of $10^2$ or greater. 
We emphasize that Fig. \ref{fig:plot} is just an illustration for a single choice of parameters, and is not intended as a comprehensive performance comparison.

\begin{figure}[h]
    \centering
    \includegraphics[width=0.8\linewidth]{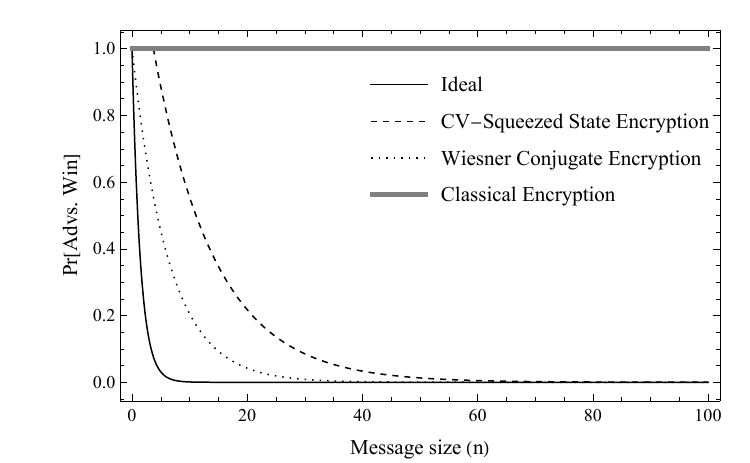}
    \caption{ \textit{Winning probabilities in the unclonability game for various schemes. The ideal curve is the simple guessing probability $\left(\frac12\right)^n$. 
    The wining probability of the conjugate coding scheme is at most $\left(\frac{1}{2} + \frac{1}{2\sqrt{2}}\right)^n$.
    For our CV scheme we plotted $2^{-n+\tau(\lambda)}$ (\ref{result}) with parameters $r=3.6, \alpha=0.4$ and allow bit flips $t=3.5\%$ of ciphertext size.} }
    \label{fig:plot}
\end{figure}

We see several avenues for improvement.
It is possible to get some extra tightness in the inequalities.
The $\qt(\ql)$ can be slightly improved by performing error correction on the 
$q$-block and $p$-block separately; this would change the $N-2t$ in (\ref{reduction}) to $N-t$.

It would be interesting to see if a tighter upper bound can be derived for the winning probability in Game~\ref{game:FullMonogamy}
(full monogamy game).
We have now used the upper bound for the {\em partial} game, and we expect that this introduces
a significant loss of tightness.

As discussed above, our results establish the security notion of \emph{unclonability}, rather than the stronger guarantee of \emph{unclonable indistinguishability}. Arguments based on monogamy of entanglement do not appear to yield guarantees beyond search security; establishing unclonable indistinguishability for CV protocols will likely require new proof techniques. We believe this is attainable and leave a full treatment to future work.

Furthermore, this work can pave the way for other schemes based on unclonability, for example, single-decryptor encryption \cite{GeorgiouZhandry2020},
or revocable commitment.

\begin{appendices}

\section{The Entangled State}
\label{app:state}
In the prepare-and-send scheme (algorithm \ref{alg:Enc}), 
we restrict $k_i$ to the interval $(-\qa\tanh r,\qa\tanh r)$. 
To ensure compatibility with the EB version, the corresponding outcome $u_i$ of the homodyne measurement 
on the Challenger's mode needs to be similarly restricted, $u_i-\qa(-1)^{c_i}\in(-\qa,\qa)$. 
Consider the Two-Mode Squeezed state without displacement
\begin{equation}
    \left| \text{EPR}_{r,\varphi}(0) \right\rangle = \frac{1}{\cosh \frac r2} \sum_{n=0}^{\infty} \left( -e^{i\varphi} \tanh \frac r2 \right)^n \left| n, n \right\rangle.
\end{equation}
In the position eigenbasis of both modes, the amplitude of this state can be represented using Hermite polynomials~$H_n$.
For $\phi=0$ we have

\begin{equation}
\langle x_1, x_2 | \text{EPR}_{r,0}(0) \rangle = \frac{e^{-\frac{1}{2} (x_1^2 + x_2^2)}}{\sqrt{\pi} \cosh \frac r2} \sum_{n=0}^{\infty} \frac{1}{n!} H_n(x_1) H_n(x_2) \left(\frac{1}{2} \tanh \frac r2 \right)^n.
\end{equation}
Using the identity
$
\sum_{n=0}^{\infty} \frac{1}{n!} H_n(x) H_n(y) \left(\frac{z}{2}\right)^n = \frac{1}{\sqrt{1 - z^2}} \exp \left[ \frac{2z}{1+z} xy - \frac{z^2}{1-z^2} (x-y)^2 \right],
$
we can simplify the expression to
\begin{equation}
    \langle x_1, x_2 | \text{EPR}_{r,0}(0) \rangle = 
    \frac{1}{\sqrt{\pi}} e^{-\frac{1}{4} e^{r} (x_1 - x_2)^2}  e^{-\frac14 e^{-r}(x_1+x_2)^2}.
\end{equation}
Now we switch to the displaced EPR state, where both modes are translated over $\pm\qa$ in the $\phi=0$ direction,
\be
    \langle x_1, x_2 | \text{EPR}_{r,0}(\pm\qa) \rangle = 
    \frac{1}{\sqrt{\pi}} e^{-\frac{1}{4} e^{r} (x_1 - x_2)^2}  e^{-\frac14 e^{-r}(x_1+x_2\mp 2\qa)^2}.
\ee
Finally we impose the restriction on the homodyne outcome by restricting the $x_1$-range,
\begin{eqnarray}
    \ket{ {\rm EPR}_{r,0}(\pm\qa) }_{\rm restr} &\propto& 
    \int_{\pm\qa-\qa}^{\pm\qa+\qa}\! \rd x_1\; \ket{x_1}
    \otimes \int_{-\infty}^\infty\!\rd x_2\; \ket{x_2}  \langle x_1, x_2 | \text{EPR}_{r,0}(\pm\qa) \rangle
    \\ &\propto & \!\!
    \int_{\pm\qa-\qa}^{\pm\qa+\qa}\!\!\!\! \rd x_1\; \ket{x_1} e^{-\frac{(x_1- \pm\qa)^2}{2\cosh r}}
    \otimes \int_{-\infty}^\infty\!\!\!\!\rd x_2\;\ket{x_2}
    e^{-\frac{\cosh r}2[x_2-\pm\qa -(x_1-\pm\qa)\tanh r ]^2}
    \quad\quad
    \\ &=&
    \int_{-\qa}^{\qa}\!\! \rd x\; \ket{x\pm\qa} e^{-\frac{x^2}{2\cosh r}}
    \otimes \int_{-\infty}^\infty\!\!\rd x'\;\ket{x'\pm \qa}
    e^{-\frac{\cosh r}2[x' -x\tanh r ]^2}.
\end{eqnarray}
It is not necessary to have a Gaussian-like distribution for the quadrature in the Challenger's mode.
Here we specified an entangled state that resembles the standard TMS state, because it has a special property:
The effect of $\ket{ {\rm EPR}_{r,\phi}(\pm\qa) }_{\rm restr}$ can be obtained 
(although inefficiently)
by taking the standard TMS state,
homodyne measuring the Ch mode, and discarding if $u_i$ does not lie in the right interval. 


\end{appendices}
\printbibliography
\end{document}